%% file: fractional_growth_portfolio_investment.tex
\documentclass[12pt]{article}
\usepackage{fullpage}
\usepackage{amssymb}
\usepackage{graphicx}
\usepackage{tocloft}
\usepackage{bm}
\usepackage{float}
\usepackage[round]{natbib}
\usepackage{authblk}

\input{defs}

\newcommand{\asconv}{\stackrel{a.s.}{\rightarrow}}
\newcommand{\diag}[1]{\mbox{diag}(#1)}
\newcommand{\ones}{\vv{e}}
\newcommand{\var}{\mbox{\bf Var}}

\author{A.E. Brockwell\thanks{anthony@brockwellanalytics.com}}
\title{Fractional Growth Portfolio Investment}

\begin{document}
\maketitle

\begin{abstract}
  We review some fundamental concepts of investment from a mathematical perspective,
  concentrating specifically on fractional-Kelly portfolios, which allocate
  a fraction of wealth to a growth-optimal portfolio while the
  remainder collects (or pays) interest at a risk-free rate.
  We elucidate a coherent continuous-parameter time-series framework for analysis
  of these portfolios, explaining relationships between Sharpe ratios, growth rates,
  and leverage.   We see how Kelly's criterion prescribes the same leverage as Markowitz mean-variance
  optimization.  Furthermore, for fractional Kelly portfolios,
  we state a simple distributional relationship between portfolio Sharpe ratio,
  the fractional coefficient, and portfolio log-returns.
  These results provide critical insight into realistic expectations of growth for different
  classes of investors, from individuals to quantitative trading operations.
  
  %For a repeatable scalable bet with known outcome distribution, Kelly's formula prescribes
  %the amount of risk that maximizes long-term growth of capital.
  %We show how Kelly's formula can be derived in a continuous-time continuous-outcome framework using
  %stochastic calculus.   Extending the results in a natural way to the multivariate case, we see that
  
  We then illustrate application of the results by analyzing performance of various
  bond and equity mixes for an investor.  We also demonstrate how the relationships can be exploited by a simple method-of-moments calculation to estimate portfolio Sharpe ratios and levels of risk deployment,
  given a fund's reported returns.

%Finally, we consider several practical applications of the theory, specifically
%targeted at the level of a mathematically-inclined home investor who wants
%to understand the nature of returns of various portfolios.  We also
%see how a simple method-of-moments approach based on the key result
%allows us to infer approximate
%Sharpe and levels of risk-deployment of a fund or trading operation that
%provides intermittent reports of returns.

\end{abstract}


\newpage

\tableofcontents

\section{Introduction}

This document describes quantitative
relationships which are important in investment,
including the concept of leverage.  In doing so, it provides
a range of insights that could be useful to home
investors, quantitative traders, or anyone who is
managing liquid investments.

%Some of these components
%have been established in existing financial literature, while some
%relationships 
%relationships appear to be undocumented in existing financial literature.
%\footnote{For example, it does not appear to be widely recognized that
%the natural multivariate extension of Kelly's formula yields
%the same mean/variance optimization problem studied by {\bf citations.}}

%Of course a home investor
%could carry out quantitative trading and a quantitative trader could
%rely on traditional forms of investment, so the line
%between the two could really be regarded as a continuum.

%\section{Problem Formulation}

We are interested in managing  a block of capital
which can be invested into
one or more risky \emph{instruments}, or held in cash/debt
at a ``risk-free'' interest rate.  An instrument is simply
a particular type of asset that can be bought and sold,
and exhibits some growth (or decay) in value over time.   
It is typical to aim for good 
long-term growth rate in total capital by investing
in a mix of different instruments.   Below we study the
question of determining the ``optimal'' mix.
By nature of the construction of the portfolio, we
ensure theoretically, that if gains/losses on capital are re-invested
continuously in time, then the pool of capital will retain non-negative
value.   In other words, we cannot lose more than our initial block
of capital. \footnote{However, it is worth noting that this is a theoretical result,
since reinvestment of profits/losses continuously in time is a practical
impossibility in most cases.}

Leverage is a critical tool in optimizing capital growth.
Roughly speaking, leverage is a mechanism
by which we can amplify the returns of an investment.
While it might seem natural that we should amplify the returns of a ``good'' investment
as much as possible, we will see that this is not a good idea.
To give a concrete example, imagine that on two successive days we observe
a 10\% loss followed by a 10\% gain.  This leaves us with
a cumulative 1\% loss ($0.9 \times 1.1 = 0.99$).  If we were to amplify the
two returns by a factor of two, however, a 20\% loss followed by a 20\% gain leaves us
with a 4\% loss ($0.8 \times 1.2 = 0.96$).  Thus, somewhat counter-intuitively, doubling the
size of our returns causes our cumulative loss to be
be \emph{worse} than double the original loss.
This asymmetry suggests to us a potential connection between
volatility, expected growth, and optimal leverage.

In the remainder of this paper, we consider this problem
from a continuous-parameter time series perspective, studying the impact of leverage analytically
using stochastic calculus.  It is beyond the scope of this document to
provide an introduction to stochastic calculus, but for a useful starting point, see,
e.g. \cite{Oksendal}, or \cite{Karatzas} for a more comprehensive treatment.
To the author's knowledge, the results stated in theorem form in this paper
do not appear in existing financial literature, although components of the
framework itself, along with associated topics of discussion, can be found in
the literature.  In particular,
\cite{breiman_gambling, hakansson_capitalgrowth, thorpe_kelly} give
much detailed discussion of Kelly portfolios and their advantages and disadvantages.
Useful related information can also be found in~\cite{kelly_collection}.
As in \cite{sid_browne1},
we will adopt a continuous-time framework
based on the use of stochastic differential equations,
to analyze a portfolio of instruments whose prices
follow a multivariate geometric Brownian motion.
%The specific problem of avoiding ``bad'' outcomes in
%such portfolios is discussed and
%addressed in \cite{maclean_growthwithsecurity} and references therein.
%These authors suggest the use of fractional Kelly schemes as a solution.
%Here we also examine this approach, but restrict attention to
%fractional Kelly portfolios and make use of
%stochastic differential equations for analysis.
%Doing so, we obtain simple yet useful ways to quantify the distribution
%of outcomes.
\cite{maclean_growthwithsecurity} (and references therein)
point out the importance of avoiding ``bad'' outcomes
in such portfolios, pointing out that fractional Kelly schemes (among
other solutions) are useful in this context.
In this paper, we also address this problem, obtaining
a specific characterization of the trade-off between growth and safety.

\section{Problem Formulation}

Suppose that, at time $t=0$, an investor has a
block of capital $A_0$ to invest, and assume
that it can be allocated to  cash, which earns interest
continuously at the \emph{risk-free rate} $r$ per unit time,
or to one or more different risky instruments
which generate returns over time.
Over time, we will denote the total capital owned by
the investor as
\be
\{A_t, ~ t \in \Rr \}.
\ee
This quantity should be thought of as the total value of
all assets, including investments in particular instruments,
and cash/debt.
We will assume the existence of a set of $m$ different
investment instruments (not including cash), which
have price time series
\be
\{P_{t,j},~t \ge 0\},~j=1,2,\ldots,m.
\ee
In what follows, we will often refer to the vector
of prices by
\be
\vv{P}_t = (P_{t,1},\ldots,P_{t,m})^T.
\ee

The collection of investments in the $m$ instruments
as well as cash will be referred to as
the \emph{portfolio}.  As mentioned above, the value of
the portfolio at time $t$ is denoted by $A_t.$
Any cash component of the portfolio receives
continuously-compounded interest at the
risk-free rate $r$.  It will be useful to define a vector
version of the risk-free rate
\be
\vv{r} = (r,r,\ldots,r)^T \in \Rr^m.
\ee

We define the \emph{leverage} vector to be
\be
\vv{k} = (k_1,\ldots,k_m)^T.
\ee
and we define \emph{total-leverage}
\be
\kappa = \sum_{j=1}^m k_j.
\ee
We will say that our portfolio is \emph{non-leveraged}
if $\kappa \le 1.$ 

At any given time $t$, the investor
invests fractions $k_j$ of his/her capital
$A_t$ into the instruments with price $P_{t,j},$
holding the remainder in cash if that remainder is positive,
or maintaining the required debt otherwise.   Cash is assumed to earn
the interest-free rate $r$, while debt pays interest at the same rate.
By construction, the amount of cash held at time $t$ is clearly
\be
A_t (1 - \kappa).
\ee

The leverage vector determines the
manner in which returns on the $m$ investment instruments are related
to returns in the total capital.  To state this
property formally, let us define
the infinitesimal return of an investment instrument by
\be
dP_{t,j} / P_{t,j}.
\ee

Then define the infinitesimal return of our total capital by
\be
dA_t / A_t.
\ee

The infinitesimal returns are related by
\be
\label{eq:levdef}
dA_t / A_t = (1-\kappa) r dt + \sum_{j=1}^m k_j dP_{t,j} / P_{t,j}
\ee
Equation~(\ref{eq:levdef}) can be viewed as a mathematical
definition of leverage.  The first term on the right represents the return
from risk-free rate interest accrual/payment on the cash/debt portion of the portfolio, while the summation
represents return contributed by the leveraged investments.

The components of our leverage vector do not necessarily
have to be less than one, or add up to one, or even be non-negative.

For example, suppose that we have $m=2$ possible investments: an S\&P500 mutual fund
and a long-term bond fund.  (We study a case like this in more detail later in
the paper.)  A leverage vector of $(0.2,0.2)$ would indicate that
at any given point in time, we keep one fifth of our total
capital in the S\&P500 fund, one fifth in the bond fund, and the
remaining three fifths held in cash, accruing interest payment at the risk-free rate.
Alternately, a leverage vector of $(0.2,1.8)$
would indicate that we would
invest one fifth of initial capital in the S\&P500 fund, borrow an amount equal to
the initial capital, and invest the four fifths remaining initial capital, along with
the borrowed sum, in the bond fund.

\subsection{Objective}

There are many possible investment objectives.   For example, one could attempt to
minimize the probability of ultimately losing all their capital, attempt to maximize the expected
return relative to some benchmark, maximize some utility function of wealth, etc.
In this paper we aim for good long-term growth profiles for
our capital $\{A_t\}$.   To state this more precisely, we will be particularly
interested in the \emph{expected log-return per unit time}
\be
L = \eof{\log(A_{t+\delta} / A_t)} / \delta
\ee
and \emph{log-return variance per unit time}
\be
V = \var(\log(A_{t+\delta} / A_t)) / \delta.
\ee
As we will see below, this approach leads us to a convenient
stochastic calculus-based derivation of Kelly's formula, and yields
additional useful insights.

Direct maximization of $L$ (with no regard for $V$)
yields the so-called \emph{growth-optimal portfolio}, also referred to as the \emph{Kelly portfolio}.
In the long-run, the growth-optimal portfolio almost-surely leads to
more wealth than any competing portfolio.  \footnote{More formally, if $\{A_t\}$ represents the capital associated with the growth-optimal portfolio,
and $\{A'_t\}$ represents the capital from another portfolio in the same instruments, chosen in a different manner,
then under a modest set of regularity conditions, we can show that $(A'_t/A_t) \asconv 0$ as $t \rightarrow \infty.$}
The price to pay for this growth dominance
however, is relatively large up and down-swings in $\{A_t\}$.
\emph{Fractional Kelly schemes},
which simply apply a fractional multiplier to the growth-optimal portfolio
leverage vector, provide one way to address this problem.

\subsection{Mathematical Analysis}

We are now in a position to analyze the impact of leverage choice.
For the sake of exposition, we address the univariate case before developing
the multivariate case.

\subsubsection{Univariate Case}

In the univariate ($m=1$) case, we can write
down a formal stochastic differential equation to describe the
price $\{P_t\}$ of an instrument over time $t \in \Rr, t \ge 0,$ as
\begin{eqnarray}
\label{eq:sde1}
dP_t & = & \mu P_t dt + \sigma P_t dW_t, \\
\label{eq:sde2}
P_0 & = & 1,
\end{eqnarray}
where $\{W_t\}$ is a standard Brownian motion.

The solution to~(\ref{eq:sde1},\ref{eq:sde2}) is a geometric Brownian motion,
for which
\be
d \log(P_t) = (\mu - \sigma^2/2) dt + \sigma dW_t.
\ee
(This follows directly from an application of It\^o's formula to~(\ref{eq:sde1},\ref{eq:sde2}).)
The process has several important properties, which we will simply state here
without proof.
\begin{enumerate}
\item It provides a realistic description of many
  real-life instruments that can be bought as investments.
  The parameters $\mu$ and $\sigma$ vary, however, over
  different investments and arguably also over time for a particular
  investment.
\item $P_t > 0.$
\item
  The log-return of the price satisfies
  \be
  \label{eq:logprice}
  \log(P_{t+\delta}/P_t) \sim \mbox{N}(\mu \delta - \sigma^2 \delta / 2, \sigma^2 \delta). 
  \ee
\end{enumerate}

Now let us consider the behavior of total capital $\{A_t,~t \ge 0\}$ over time
when we invest a fraction $k$ (leverage) in the instrument whose price $P_t$ is
given by~(\ref{eq:sde1},\ref{eq:sde2}).
It follows directly from~(\ref{eq:levdef}) that
\begin{eqnarray}
  dA_t & = & (1-k) r A_t dt + k A_t dP_t / P_t \\
       & = &  [(1-k) r + k \mu] A_t dt + k A_t \sigma dW_t.
\end{eqnarray}
In other words, like the underlying instrument price, our amount of capital $A_t$ also
follows a geometric Brownian motion process, but with different parameters.   Consequently,
\be
\label{eq:dat1}
d \log(A_t) = [(1-k) r + k \mu - k^2 \sigma^2/2] dt + k \sigma dW_t.
\ee
It is very important to note\footnote{It is a common mistake to presume, for example, that a triple-leverage investment yields triple the log-returns over time.  This is not the case.}
that
\be
d \log(A_t) \ne k \cdot d\log(P_t).
\ee

It follows directly from~(\ref{eq:dat1}) that
\begin{enumerate}
\item $A_t > 0.$   This is our guarantee that we do not lose more than our
  initial capital.   However, note that it relies on the unrealistic assumption
  that we are able to carry out continuous reinvestment of profits/losses.
\item
  The log of total amount of capital (including reinvested profits/losses) satisfies
  \be
  \label{eq:sde3}
  \log(A_{t+\delta} / A_t) \sim \mbox{N}(r \delta + k (\mu-r) \delta - k^2 \sigma^2 \delta / 2, k^2 \sigma^2 \delta).
  \ee
\end{enumerate}

At any time $t$, equation~(\ref{eq:sde3}) provides a predictive distribution for our amount of capital
at a time point $(t+\delta)$ in the future.   Specifically, it is log-normal.

Following on from~(\ref{eq:sde3}), we see that the \emph{expected log-return per unit time} of
our capital is
\be
\label{eq:elrt} 
L(k) = \eof{\log(A_{t+\delta} / A_t)} / \delta = r + k (\mu-r) - k^2 \sigma^2 / 2.
\ee
Differentiating this with respect to $k$ and setting to zero gives us
a formula for the value of leverage $k$ that maximizes
expected log-return per unit time.   This is simply
\be
\label{eq:kelly}
k^* = (\mu-r) / \sigma^2,
\ee
and the corresponding expected log-return per unit time is
\be
\label{eq:sharpelrt}
L(k^*) = r + \frac{1}{2} (\mu-r)^2/\sigma^2.
\ee
The expression~(\ref{eq:kelly}) is well-known, and is typically referred to
as \emph{Kelly's formula}, for Kelly's analysis of a closely-related
problem (see ~\cite{kelly}).

\subsubsection{Multivariate Case}
\label{sec:mv_kelly}

It is straightforward to extend both~(\ref{eq:elrt}) and~(\ref{eq:kelly}) to
the multivariate case when we have multiple investments ($m>1$), and we allocate
proportions $\vv{k} = (k_1,k_2,\ldots,k_m)^T$ of capital to the respective
investments.
In this case, we define $\vv{\mu} = (\mu_1,\ldots,\mu_m)^T,$
$\vv{\sigma} = (\sigma_1,\ldots,\sigma_m)^T,$ as well as a correlation
matrix $R \in \Rr^{m \times m}.$  The multivariate analog of~(\ref{eq:sde1})
becomes
\be
\label{eq:mvsde0}
d\vv{P}_t = \diag{\vv{\mu}} \vv{P}_t dt + \diag{\vv{\sigma}} \diag{\vv{P}_t} d\vv{U}_t,
\ee
where $\{\vv{U}_t\}$ is a multivariate Brownian motion with correlation
matrix $R$, so that
\be
\eof{\vv{U}_t} = 0, \quad \var(\vv{U_t}) = Rt,
\ee
and $\diag{\cdot}$ represents a square matrix with diagonal
elements given by the vector argument, and zeros in all off-diagonal positions.
It will be convenient to define
\be
\Sigma = [\diag{\vv{\sigma}}] R [\diag{\vv{\sigma}}].
\ee
Also, recall the definition $\kappa = \sum k_j = \vv{k} \cdot \vv{e}_m,$
where $\vv{e}_m = (1,1,\ldots,1)^T \in \Rr^m.$

Applying~(\ref{eq:levdef}), it is straightforward to show that there is another
standard Brownian motion $\{W_t\}$ such that the total capital $\{A_t\}$
satisfies
\begin{eqnarray}
  dA_t & = & ((1-\vv{k}\cdot \vv{e}_m)r + \vv{k} \cdot \vv{\mu}) A_t dt + (\vv{k}^T \Sigma \vv{k})^{1/2} A_t dW_t \\
       & = & (r - \vv{k} \cdot \vv{r} + \vv{k} \cdot  \vv{\mu}) A_t dt + (\vv{k}^T \Sigma \vv{k})^{1/2} A_t dW_t,
\end{eqnarray}
and it follows directly from application of It\^o's formula that
\be
\label{eq:mvsde}
d\log(A_t) = (r + \vv{k} \cdot (\vv{\mu} - \vv{r}) - \vv{k}^T \Sigma \vv{k} /2) dt
   + (\vv{k}^T \Sigma \vv{k})^{1/2} dW_t,
\ee
so the expected log-return per unit time is given by
\be
\label{eq:mvelrt}
L(\vv{k}) = \eof{\log(A_{t+\delta} / A_t)} / \delta = r + \vv{k} \cdot (\vv{\mu}-\vv{r}) - \vv{k}^T \Sigma \vv{k} / 2,
\ee
Consequently Kelly's formula becomes
\be
\label{eq:mvkelly}
\vv{k}^* = \Sigma^{-1} (\vv{\mu} - \vv{r}),
\ee
with maximum expected log-return per unit time
\be
\label{eq:mvsharpelrt}
L(\vv{k}^*) = r + \frac{1}{2} (\vv{\mu}-\vv{r})^T \Sigma^{-1} (\vv{\mu}-\vv{r}).
\ee
We will also be interested in the variance of the log-return per unit time.
From~(\ref{eq:mvsde}), this is easily seen to be
\be
\label{eq:mvvlrt}
V(\vv{k}^*) = \var(\log(A_{t+\delta} / A_t)) / \delta
 = \vv{k}^T \Sigma \vv{k}.
 \ee

 Readers may note that~(\ref{eq:mvkelly}) is recognizable as the solution of a
 standard Markowitz mean-variance portfolio
 weight selection problem \cite[see][]{markowitz_meanvariance}.

\subsection{Sharpe Ratios and Fractional Kelly Investment} 

In the finance industry, in the univariate case,
it is common practice to refer to
the quantity $(\mu-r)/\sigma$, as the \emph{Sharpe ratio} of the investment.
\footnote{It is important to differentiate between
  this quantity, and estimators thereof.
  Unfortunately, people often use the same term to refer to both
  the quantity and its estimator(s).}
More generally, we will define the \emph{Sharpe ratio} of
a single investment or a set of investments as
\be
\label{eq:sharpedefn}
S = [(\vv{\mu}-\vv{r})^T \Sigma^{-1} (\vv{\mu}-\vv{r})]^{1/2},
\ee
where $\vv{\mu}$ and $\Sigma$ are the parameters defining the
growth and volatility of the component investments, and $\vv{r}$ is the
vector whose elements are all equal to the risk-free rate.
(Note that the one-dimensional multivariate expression is consistent with the univariate definition.)

Sharpe ratios are clearly important
because the maximum growth rates in~(\ref{eq:sharpelrt}, \ref{eq:mvsharpelrt})
increase quadratically with respect to the Sharpe ratio.

We have seen that Kelly's formula prescribes the leverage that
maximizes expected log-return per unit time $L$ of an investment.
However, to limit up/down-swings in capital,
it is desirable to operate with
a lower total level of risk.  To achieve this, we can apply a multiplier to
~(\ref{eq:mvkelly}), 
\be
\label{eq:frack}
\vv{k}_\alpha = \alpha \Sigma^{-1} (\vv{\mu} - \vv{r}), \quad \alpha \in [0,1].
\ee
This quantity is often referred to as \emph{fractional Kelly} leverage.
In this light, fractional Kelly investment can be viewed
as a particular means of leverage selection, with $\alpha$ specifying the level of risk relative to
that required for the growth-optimal portfolio.

The following result makes a formal connection
between leverage, log-returns and Sharpe ratios.

\begin{theorem}[Sharpe-Leverage Performance Profile]
  \label{th:slpp}
  Suppose that we apply fractional Kelly leverage~(\ref{eq:frack}) with
  multiplier $\alpha \in [0,1]$ to the
  portfolio with prices governed by~(\ref{eq:mvsde0}),
  earning/paying the risk-free interest rate $r$
  on the cash/debt component.
  Then the resulting capital process $\{A_t\}$ is a geometric Brownian motion
  satisfying
  \be
  \label{eq:agbm}
d\log(A_t) = [r + (\alpha - \alpha^2/2) S^2] dt 
   + \alpha S dW_t,
   \ee
where $S = [(\vv{\mu}-\vv{r})^T \Sigma^{-1} (\vv{\mu}-\vv{r})]^{1/2}$ denotes the Sharpe ratio
of the portfolio.   
Consequently, $\{A_t\}$ has expected log-return per unit time
\be
\label{eq:fracelrt}
L(\vv{k}_\alpha) =   \eof{\log(A_{t+\delta} / A_t)} / \delta =
r + (\alpha - \alpha^2/2) S^2
\ee
and log-return variance per unit time
\be
\label{eq:fracvlrt}
V(\vv{k}_\alpha) = \var(\log(A_{t+\delta}/A_t))/\delta = \alpha^2 S^2.
\ee
\end{theorem}

\begin{proof}
  The derivation of~(\ref{eq:mvsde}) above establishes that $\{A_t\}$ is a geometric Brownian
  motion.  Substituting~(\ref{eq:frack}) into~(\ref{eq:mvsde}) and making use of
  the definition~(\ref{eq:sharpedefn}) yields~(\ref{eq:agbm}).
  The remaining results then follow directly.
\end{proof}

When the conditions of Theorem~\ref{th:slpp} are met and $\alpha = 1,$ we
obtain the growth-optimal portfolio.   More generally,
equations~(\ref{eq:fracelrt}, \ref{eq:fracvlrt}) show
the trade-off we obtain with different values of $\alpha.$   Values closer to one yield higher expected log-return
per unit time on capital, but values closer to zero give better (reduced) variance
of the log-return around the mean.

\section{Examples and Applications}

We now consider some practical implications of the theory outlined
above, beginning with analysis that would be of interest to a
typical individual with some savings to invest.

\subsection{Equity and Bonds Mix}

Let us consider a portfolio consisting of only two instruments:
$\{P_{t,1}\}$ will represent the price of Vanguard 500 Mutual Fund Index shares
(VFIAX), and $\{P_{t,2}\}$ will
represent the price of Vanguard Long-Term Treasury Fund shares (VUSUX).

\begin{figure}
  \includegraphics[width=7in]{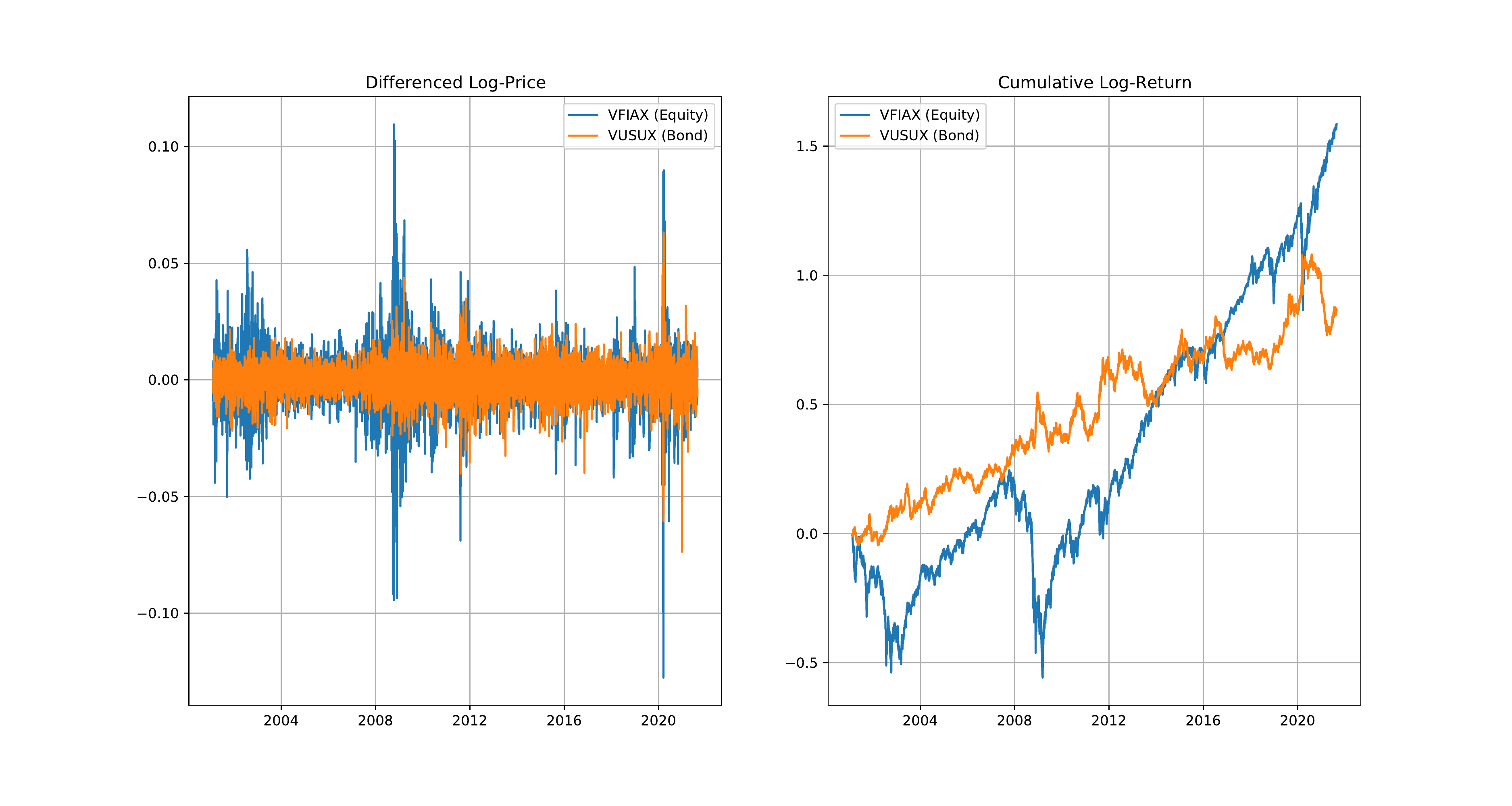}
  \centerline{\parbox[t]{5.5in}{
  \caption{Equity (VFIAX) and Bond (VUSUX) performance from Feb. 12th, 2001 to Aug. 24th, 2021.
    Data source: Yahoo Finance.}
  \label{fig:common}
   }}
\end{figure}

Figure~\ref{fig:common} shows differences of logarithms of adjusted daily closing prices for these
two instruments from 2001 to 2021.   This time period covers the dot-com bubble of 2002 and the sub-prime crisis of 2008,
as well as more recent market turbulence at the onset of the coronavirus pandemic in early 2020.

In what follows, we will assume that the risk-free rate $r$ is equal to zero.
This is a fairly reasonable approximation for the post-sub-prime period from 2008 to 2021,
although results here can be easily adapted if $r$ is assumed to be non-zero.

\subsubsection{Parameter Estimation}

From the raw data shown in the plots in Figure~\ref{fig:common},
we can easily compute estimators of $\vv{\mu} = (\mu_1,\mu_2)$,
$\vv{\sigma} = (\sigma_1,\sigma_2)$ and $R$.
Let the unit of time be one year, assume there are
260 trading days each year, and define the daily log-returns
(shown in the left part of Figure~\ref{fig:common})
as
\be
D_{t,j} = \log(P_{(t+1)/260,j}) - \log(P_{t/260,j}), \quad t=1,2,\ldots,n-1,
\ee
where $n$ is the total number of trading days of data.
We know from~(\ref{eq:logprice}) that
\be
\label{eq:est1}
D_{t,j} \sim \mbox{N}((\mu_j - \sigma^2_j/2)/260, \sigma^2/260).
\ee
Furthermore,
\be
\mbox{Corr}(D_{t,1},D_{t,2}) = R_{12}.
\label{eq:est2}
\ee
We can therefore take sample mean and variance from
our two daily price series, compute sample correlation,
and use them to construct method-of-moments estimators
\begin{eqnarray}
  \hat{\sigma}^2_j & = & {260 \over n-2} \sum_{t=1}^{n-1} (D_{t,j} - \bar{D}_j)^2,
  \quad \bar{D}_j = {1 \over n-1} \sum_{t=1}^{n-1} D_{t,j}, \\
  \hat{\sigma}_j & = & \sqrt{\hat{\sigma}^2_j}, \\
  \hat{\mu}_j & = & 260 \bar{D}_j + \hat{\sigma}^2_j / 2, \\
  \hat{R}_{jk} & = & {260 \over \hat{\sigma}_j \hat{\sigma}_k (n-2)} \sum_{t=1}^{n-1} (D_{t,j} - \bar{D}_j)(D_{t,k} - \bar{D}_k)
\end{eqnarray}
With the raw data described above, we obtain
\be
\label{eq:estimates1}
\hat{\vv{\mu}} = 
(0.099, 0.051)^T, \quad
\hat{\vv{\sigma}} = (0.199, 0.123)^T, \quad
\hat{R}_{12} = -0.377.
\ee

From a practical perspective, we also want to consider the different
impact of taxation on these two instruments.   In the United States, tax on
a mutual fund that tracks the S\&P500 is typically close to
the long-term capital gain rate of 20\%, while tax on bond funds,
although more complex to compute, is typically closer to standard
income tax, which is often around 40\%.   To approximate the impact of taxation,
we will therefore adjust the estimates~(\ref{eq:estimates1}) above by multiplying components
of $\mu$ by one minus the corresponding tax rate, giving us final
estimates
\be
\label{eq:estimates2}
\hat{\vv{\mu}} = (0.079, 0.031), \quad
\hat{\Sigma} = \machrix{\hphantom{-}0.0396}{-0.0093}{-0.0093}{\hphantom{-}0.0152} .
\ee

From these estimates, in turn, we can determine the Kelly leverage
for the growth-optimal portfolio
by applying~(\ref{eq:mvkelly}), obtaining
\be
\label{ex:fullkelly}
\vv{k}^* = \hat{\Sigma}^{-1} \hat{\vv{\mu}} \simeq (2.89, 3.78)^T.
\ee
If $\vv{\mu}$ and $\Sigma$ were indeed the same as these estimates,
and if we could borrow money
to provide leverage, this says we could maximize growth over time
of capital by borrowing $5.67$ times our initial capital (assuming a no-interest loan) to
obtain total leverage of $2.89+3.78=6.67$.  We would then invest the original capital along with the borrowed capital
into the two components of the portfolio, rebalancing to hold the fraction of capital in the two components constant.
The resulting expected annual log-return on initial (non-borrowed) capital would
be
\be
\label{eq:fullkellysp}
\vv{k}^* \cdot \hat{\vv{\mu}} - {\vv{k}^*}^T \hat{\Sigma} \vv{k}^* / 2 \simeq 0.172
\ee
In this argument we have pretended that there is no estimation
error.   In fact, that source of noise can have a significant impact, but
that analysis is beyond the scope of this paper.

\subsubsection{Leverage Value Configurations}

We now consider a range of different investment possibilities.
Some of these will be applied to only one of our two instruments, and for these
cases we will use $m=1$ along with the corresponding (marginal)
components of $\hat{\vv{\mu}}$, $\hat{\Sigma}$ and the corresponding
estimate $\hat{S}$ of Sharpe ratio.
In the more interesting cases we will examine portfolios consisting of
both instruments.
We consider the following cases.

\begin{table}[H]
  \begin{center}
    \begin{footnotesize}
    \begin{tabular}{|l|l|l|l|l|r|}
      \hline
      & Leverage    & \multicolumn{3}{l|}{Drift/Diffusion}  & Kelly-Fraction\\
      \hline
      Name & $k$ & $\hat{\mu}$ & $\hat{\Sigma}$ & $\hat{S}$ & $\hat{\alpha}$ \\
      \hline
      \multicolumn{6}{|c|}{Single-instrument ($m=1$) portfolios} \\
      \hline
      Non-Leveraged Equity   & 1.00  &  0.079 & $0.199^2$ & 0.398 & 0.502 \\
      Double-Equity          & 2.00  &  0.079 & $0.199^2$ & 0.398 & 1.003 \\
      Triple-Equity          & 3.00  &  0.079 & $0.199^2$ & 0.398 & 1.505 \\
      Non-Leveraged Bonds    & 1.00   &  0.031 & $0.123^2$ & 0.252 & 0.488   \\
      Double-Bonds           & 2.00  &  0.031 & $0.123^2$ & 0.252 & 0.976 \\
      \hline
      \multicolumn{6}{|c|}{Two-instrument ($m=2$) portfolios} \\
      \hline
      Fractional Kelly(0.30) & (0.87,1.13) &  (0.079,0.031) & $\hat{\Sigma}$~(eqn~\ref{eq:estimates2}) & 0.588 & 0.30 \\
      Constr. Kelly(2.0)     & (1.33,0.67) &  (0.079,0.031) & $\hat{\Sigma}$~(eqn~\ref{eq:estimates2}) & 0.588 & NA \\
      Full-Kelly             & (2.89,3.78) &  (0.079,0.031) & $\hat{\Sigma}$~(eqn~\ref{eq:estimates2}) & 0.588 & 1.0 \\
      \hline
    \end{tabular}
    \end{footnotesize}
  \end{center}
  \centerline{\parbox[t]{5.5in}{
      \caption{Example portfolios. }
  \label{tab:exampleconfigs}
  }}
\end{table}

Cases are constructed as follows.
\begin{itemize}
\item \emph{Non-leveraged/Double/Triple Equities/Bonds:} The non-leveraged approach invests all capital directly into the
  equities/bonds instrument.  Double and triple versions invest all capital in the same instrument but with leverage $2$ or $3$, respectively.
\item \emph{Fractional Kelly (0.30):} This applies leverage $\alpha = 0.30$ times the full-Kelly leverage.
  In this example, it results in total
  leverage $\kappa$ approximately equal to $2.0.$
\item \emph{Constrained Kelly (2.0):} Here we choose $\vv{k}$ so as to maximize expected growth $L$ subject to
  the constraint that $\kappa = 2.$  This can be done using Lagrange multipliers, as described in Appendix~\ref{app:lagrange}.
\item \emph{Full-Kelly:} This choice, given by~(\ref{ex:fullkelly}) leads to maximum growth rate over time, at the cost of high variance and draw-downs.
\end{itemize}

Table~\ref{tab:exampleconfigs} also shows a term
\be
\hat{\alpha} = \vv{k} / (\hat{\Sigma}^{-1} \hat{\vv{\mu}}),
\ee
which is only defined when the leverage vector $\vv{k}$ is a multiple
of $\hat{\Sigma}^{-1}\hat{\vv{\mu}}$, and in that case represents an
estimate of the ratio of portfolio leverage to growth-optimal leverage.
Values less than $1.0$ are desirable in this context, since it is obvious from~(\ref{eq:fracelrt}) and~(\ref{eq:fracvlrt})
that $\alpha>1$ implies we are taking unnecessary extra risk for a return that can also be achieved with $\alpha<1.$

\subsubsection{Simulation Results}

It is straightforward to simulate hypothetical capital invested in the portfolios corresponding to the test cases.
Using data ranging from Feb 13th, 2001 to March 18th, 2021,
we simply compute the daily portfolio returns using the
recorded data.  In order to approximate the tax-effect on the two components,
we also apply a negative drift by subtracting the tax adjustment used on~$\vv{\mu}$ in~(\ref{eq:estimates2})
on a daily basis from the log-returns of the data.

\begin{figure}[h]
  \includegraphics[width=7in]{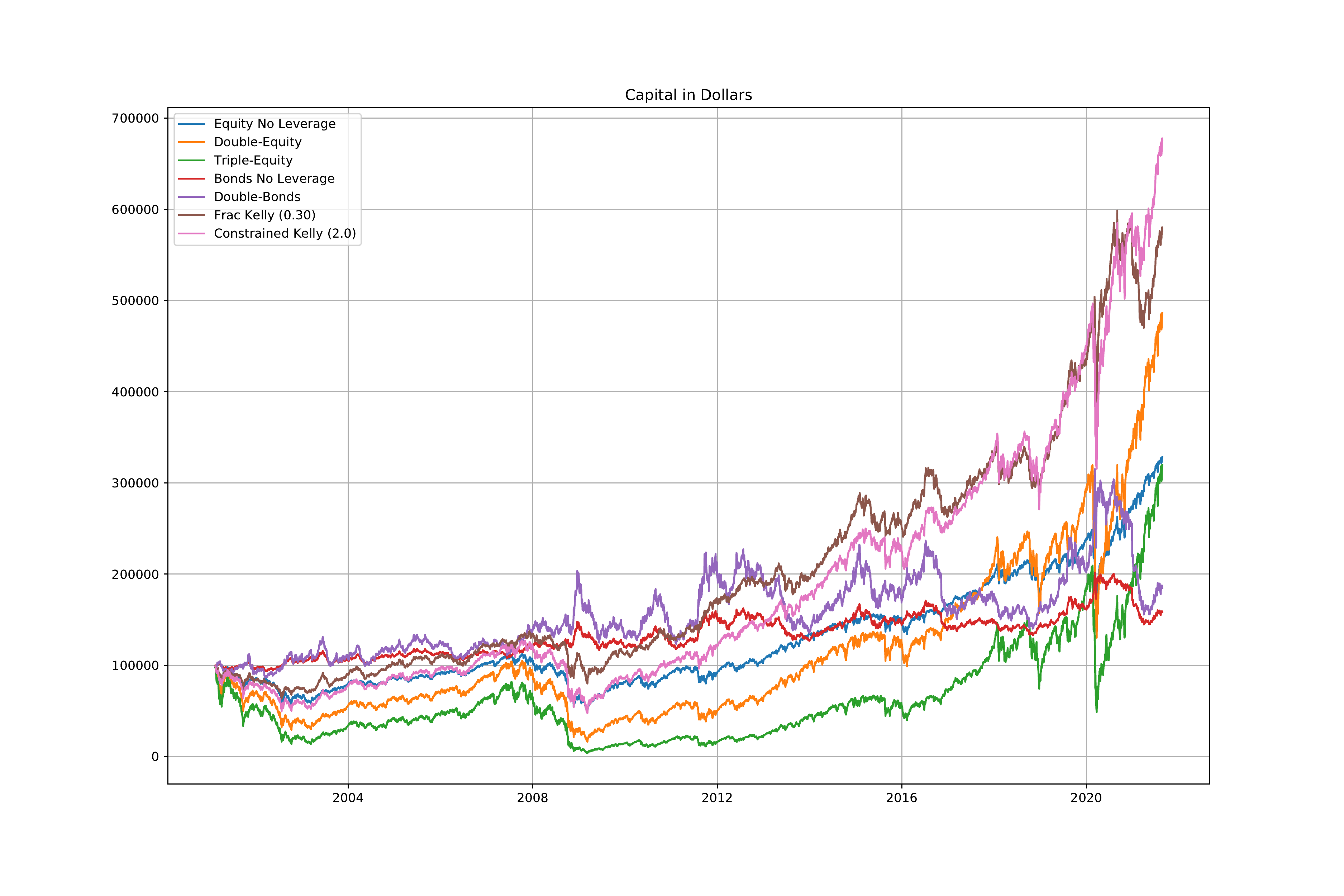}
  \centerline{\parbox[t]{5.5in}{
      \caption{Simulated post-tax capital of hypothetical portfolios built of VFIAX and VUSUX,
        assuming starting capital of \$100000, for a range of different leverage vectors. 
        \label{fig:doublelev}
  }}}
\end{figure}

Figure~\ref{fig:doublelev} shows cumulative capital $\{A_t\}$ over time for a range of different leverage settings.   The full-Kelly
leverage case is omitted from the plot since it distorts the scale, but the full-Kelly case is included in Figure~\ref{fig:doublelog}, which shows the log of cumulative capital.

\begin{figure}[ht]
  \includegraphics[width=7in]{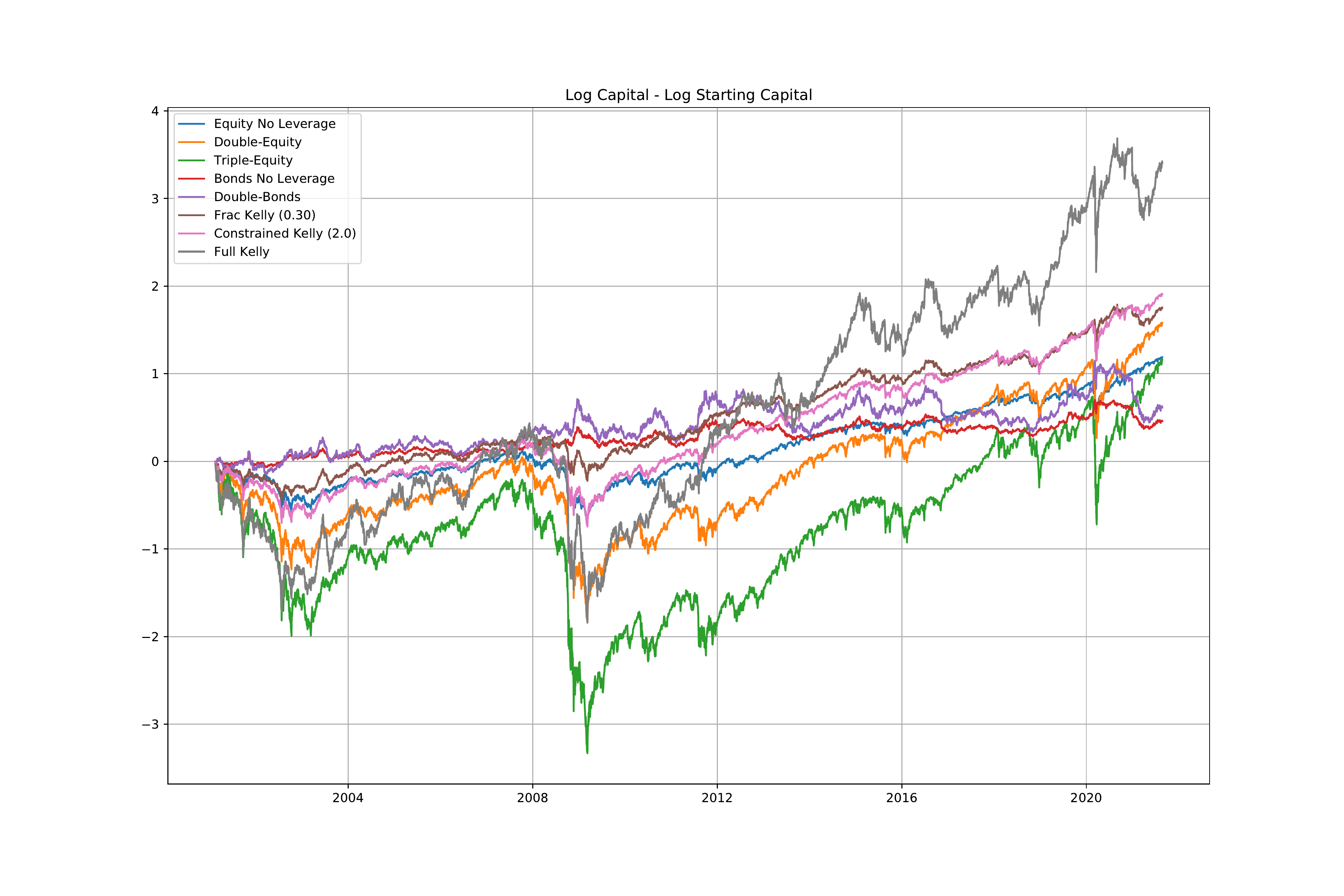}
  \centerline{\parbox[t]{5.5in}{
      \caption{Simulated post-tax log-return of hypothetical portfolios built of VFIAX and VUSUX, for a range of different leverage vectors.
        The full-Kelly case clearly dominates the others in terms of growth, but it also loses 87.7\% of its capital value from Dec. 5th, 2007 to Mar. 6th, 2009 during the sub-prime crisis.
        The simple triple-leverage equity investment also stands out with its low return and 96.3\% loss from  Feb 14, 2001 to Mar 6, 2009.
        \label{fig:doublelog}
  }}}
\end{figure}

\begin{table}[H]
  \begin{center}
    \begin{footnotesize}
    \begin{tabular}{|l|l|l|l|l|l|r|}
      \hline
         & Growth          &                 & \multicolumn{3}{l|}{Draw-down} & \\
      \hline
      Name & $\hat{L}$ & ${\hat{V}}^{1/2}$ & Max. & Start & End  & Final Value \\
      \hline
      \multicolumn{7}{|c|}{Single-instrument ($m=1$) portfolios} \\
      \hline
      Non-Leveraged Equity   &  0.060 & 0.200 & 56.4\% & Oct 8 2007 & Mar 6 2009 & \$328,143 \\
      Double-Equity          &  0.081 & 0.400 & 84.6\% & Oct 8 2007 &  Mar 6 2009 & \$486,418 \\
      Triple-Equity          &  0.060 & 0.603 & 96.4\% & Feb 14 2001 & Mar 6 2009 & \$319,574\\
      Non-Leveraged Bonds    &  0.023 & 0.123 & 28.3\% & Mar 6 2020 & Mar 30 2021 & \$158,109\\
      Double-Bonds           &  0.031 & 0.247 & 50.7\% & Mar 6 2020 & Mar 30 2021 & \$184,841\\
      \hline
      \multicolumn{7}{|c|}{Two-instrument ($m=2$) portfolios} \\
      \hline
      Fractional Kelly(0.30) &  0.089 & 0.176 & 41.9\% & Dec 5 2007 & Mar 6 2009 & \$576,464 \\
      Constr. Kelly(2.0)     &  0.097 & 0.240 & 62.7\% & Oct 8 2007 & Mar 6 2009 & \$675,775 \\
      Full-Kelly             &  0.172 & 0.594 & 89.8\% & Dec 5 2007 &  Mar 6 2009 & \$3,002,829 \\
      \hline
    \end{tabular}
    \end{footnotesize}
  \end{center}
  \centerline{\parbox[t]{5.5in}{
      \caption{Summary of simulated performance of example portfolios derived from the VFIAX and VUSUX instruments, simulated from Feb. 12th, 2001 to Aug. 24th, 2021, with initial capital $\$100,000.$  Risk-free interest rate is assumed to be zero.   $\hat{L}$ denotes the sample average of the annualized log-return, ${\hat{V}}^{1/2}$ denotes the annualized sample standard deviation of those log-returns, and the maximum \emph{draw-down} is defined as the maximum drop from peak to trough with trough occurring after the peak. }
  \label{tab:examplesummary}
  }}
\end{table}

Table~\ref{tab:examplesummary} shows corresponding summary statistics, including the sample
(annualized) growth rate $\hat{L}$, the corresponding sample variance $\hat{V}$,
as well as the final value and  maximum draw-down over the time period, which is defined as the maximum
peak-to-trough drop such that the trough occurs after the peak.  From a psychological perspective,
investors pay significant attention to difference between the
historical maximum and the current value of an investment.  Large values of this ``draw-down'' typically induce great anxiety.
Several important observations can be made from the results.   First, the full-Kelly case and the triple-equity case both exhibit very large maximum draw-downs over this time period - 89.8\% and 96.4\%, respectively.
Few investors would hold onto an investment after losses that large.
Both the fractional Kelly and the constrained Kelly options appear to give a fairly good compromise between greater growth and higher variance/draw-downs.  Of course different investors may different preferences,
but it is notable that of the choices in Table~\ref{tab:examplesummary}, apart from the non-leveraged bond case,
the fractional Kelly option has the smallest value of $\hat{V}$ and still yields a respectable growth rate $\hat{L}.$

\subsection{Application to Fund Evaluation}
%Renaissance Medallion Fund}

Hedge fund managers generally release very little information about the nature of their trading operations.
However, they often release yearly or quarterly statments of actual returns on their hedge funds.
We have already seen that Theorem~\ref{th:slpp} assists in choosing a desirable leverage vector.
Here we see how it leads to a simple method-of-moments approach to estimate portfolio Sharpe ratio and risk deployment of a fund,
using only such publicly available return data.
The approach relies on an implicit assumption that
the fund is applying a fractional Kelly approach, or something similar.  In
light of the previous observation at the end of Subsection~\ref{sec:mv_kelly}
that fractional Kelly leverage is equivalent to industry-standard Markowitz mean-variance optimization,
this is not an unrealistic assumption.
This ``reverse-engineering'' allows us to evaluate funds in
a more nuanced way than simple inspection of past returns.
For example, we can easily differentiate between funds that obtain high returns by
taking excessive risk, and those that obtain high returns by deployment of high-Sharpe ratio
portfolios.

Consider, for example, the annual returns of the well-known (and very strongly performing)
Renaissance Medallion fund.   These returns are publicly available in Appendix 1 of \cite{zuckerman_jimsimons},
and can easily be translated into log-returns and viewed in light of equations~(\ref{eq:fracelrt}) and~(\ref{eq:fracvlrt}).
Without re-printing the full table itself, we give summary statistics of their resulting (before-fee) log-returns
in Table~\ref{tab:renaissance} below.

\begin{table}[H]
  \begin{center}
  \begin{tabular}{|l|l|}
    \hline
    Date Range & 1988-2018 inclusive \\
    \# Data Points & 31 \\
  Average Log-Return & 0.490 \\
  Std. Dev. (Log-Return) & 0.187 \\
  \hline
  \end{tabular}

    \end{center}
  \centerline{\parbox[t]{5.5in}{
      \caption{Summary statistics of Renaissance Medallion Fund returns as given in \cite{zuckerman_jimsimons}.}
  \label{tab:renaissance}}
  }

\end{table}
  
Re-arranging the pair of equations~(\ref{eq:fracelrt},\ref{eq:fracvlrt}), we see that for
a fractional Kelly investor,
\begin{eqnarray}
  \alpha & = & 2V (2L + V)^{-1} \\
  S^2 & = & \alpha^{-1} (L+V/2).
\end{eqnarray}
By simply matching $L = 0.490$ and $V = 0.187^2$ from the table above, we see that the Medallion returns
are consistent with deployment of a portfolio with Sharpe ratio $S \simeq 2.72$ at approximately $\alpha=0.068$ times
the Kelly-optimal leverage.
The (estimated) value of $0.068$ places them comfortably below the obvious danger point.  A value equal to $1.0$ would indicate deployment
of risk at the growth-optimal level, with its extreme volatility.  Such a value is not palatable to a typical investor.
Any value larger than $1.0$ would indicate a
sub-optimal deployment of risk.   A value larger than $2.0$ would indicate that the fund will likely eventually collapse
in the sense that its value will converge in probability to zero.

\section{Concluding Remarks}

Given a portfolio of instruments with risk, some capital $A_0,$ and a risk-free interest rate $r$,
we have seen that the Kelly portfolio maximizes long-term growth of capital, but with distressingly large
draw-downs along the way.  This undesirable quality can be mitigated by the use of
fractional Kelly portfolios, at the cost of a reduction in long-term growth rate.
In Theorem~\ref{th:slpp}, we have directly quantified this distributional relationship,
and given explicit expressions
that help us to find a desirable balance between growth and variance of returns.
We have also seen how, under a set of reasonable assumptions,
fractional Kelly investment is the same as Markowitz mean-variance optimization.

The stochastic differential equation methodology used in this paper has been used by
others in the same context, notably \cite{sid_browne1}, but appears not to be widely appreciated.
Within this framework, one could potentially
generalize the geometric Brownian motion price models~(\ref{eq:sde1}, \ref{eq:mvsde0}) to
jump-diffusion models, thereby allowing for the more realistic case of skewed and
heavy-tailed log-returns.   Using multivariate stochastic calculus,
it would also possible to carry out analysis of situations where
drift and/or diffusion coefficients $\vv{\mu}$ and $\Sigma$ are not known
and must be estimated.

On the practical side, the relationships given
in Theorem~\ref{th:slpp} have broad applicability.
They apply equally well to
simple investment of an individual's personal funds
as to the more sophisticated operations of a quantitative trading organization or other fund.
%For a specified Sharpe ratio, one can easily determine the distribution
%of annual returns of investments as a function of the chosen level of risk deployment.
%There are many immediate applications of the basic result.
For example, individual investors can use the implied return distributions to assess the quality of
leveraged mutual funds under various assumptions.
At the more sophisticated end of the spectrum, a quantitative trading operation with a porfolio of known Sharpe ratio could
use the result to determine how much capital/leverage to deploy,
while remaining within distributional constraints imposed by their investors.
More generally, the stochastic differential equation framework
%
%Indeed, it is the author's belief that, due to its inherent ability to capture time series behaviour,
%
provides a solid foundation with which to address important
yet-unsolved practical investment problems.

\section{Acknowledgements}

The author is grateful to Peter Brockwell, Jonathan Baxter, and Peter Dodds for their
valuable comments and suggestions.
This work has benefited substantially from their feedback.

\appendix

\section{Supporting Results}

\subsection{Constrained Maximization of Expected Log-Return per Unit Time}
\label{app:lagrange}

Suppose we wish to choose the leverage vector
$k \in \Rr^m$ to maximize expected log-return per unit time
(restating equation~(\ref{eq:mvelrt}))
\be
\label{eq:app1}
L(\vv{k}) = \vv{k} \cdot \vv{\mu} - \frac{1}{2} \vv{k}^T \Sigma \vv{k},
\ee
subject to the constraint that total-leverage is
\be
\label{eq:app2}
\kappa = \sum_{j=1}^m k_j = \vv{k} \cdot \ones_m = \kappa_0,
\ee
where $\ones_m = (1,\ldots,1)^T \in \Rr^m.$
We can express the constraint as
\be
\label{eq:app4}
(\vv{k} \cdot \ones_m - \kappa_0) = 0.
\ee
Thus the Lagrangian is
\be
{\cal L}(\vv{k}) = \vv{k} \cdot \vv{\mu} - \frac{1}{2} \vv{k}^T \Sigma \vv{k} - \lambda (\vv{k} \cdot \ones_m - \kappa_0),
\ee
Differentiating with respect to the vector $\vv{k}$, and setting to zero, we find that the
value maximizing~(\ref{eq:app1}) takes the form
\be
\label{eq:app3}
\vv{k} = \Sigma^{-1} (\vv{\mu} - \lambda \ones_m).
\ee
Substituting~(\ref{eq:app3}) into the constraint~(\ref{eq:app4}),
we can solve for $\lambda$, obtaining
\be
\label{eq:app5}
\lambda = (\ones_m^T \Sigma^{-1} \vv{\mu} - \kappa_0)(\ones_m^T \Sigma^{-1} \ones_m)^{-1}.
\ee
Equations~(\ref{eq:app3}) and~(\ref{eq:app5}) together specify the leverage
vector maximizing $L$ subject to the required constraint.

\bibliography{fractional_bib}

\end{document}

%% file: defs.tex
\newcommand{\be}{\begin{equation}}
\newcommand{\ee}{\end{equation}}

\newtheorem{theorem}{Theorem}[section]

\newcommand{\fivevect}[5]
{\left(
   #1, #2, #3, #4, #5
 \end{array}\right)}

\newcommand{\vv}[1]
{
{\bf #1}
}
\newcommand{\eof}[1]
{{{\bf E} \left[ #1 \right]}}

\newcommand{\machrix}[4]
{
\left[
\begin{array}{cc}
#1 & #2 \\
#3 & #4
\end{array}
\right]
}

\newcommand{\Rr}
{{\mathbb R}}

\newenvironment{proof}[1][]
{{\bf Proof#1:} 
\hspace{0.3cm}}
{\hfill $\Box$
\vskip 1cm}

\newcounter{examplenum}[section]